\documentclass[reprint, amsmath,amssymb,aps,prx]{revtex4-2}
\usepackage[dvipsnames]{xcolor}
\usepackage{hyperref}
\usepackage{cleveref}
\usepackage{subfiles}
\usepackage{graphicx}
\usepackage{tikz-cd}
\usepackage{float}
\usepackage{titletoc}
\usepackage[utf8]{inputenc}
\usepackage{physics}
\usepackage{hyperref}
\usepackage{soul}
\usepackage{amsthm}
\newcommand{\vect}[1]{\boldsymbol{#1}}
\DeclareMathOperator*{\argminG}{arg\,min}
\usepackage{listings}
\usepackage{hyperref}
\usepackage{amsthm}
\usepackage[thinc]{esdiff}
\usepackage{tikz}       
\usetikzlibrary{fit}
\usepackage{nicematrix} 
\usetikzlibrary{quantikz}

\usepackage{amssymb,amsmath,amsthm}
\newtheorem{theorem}{Theorem}
\newtheorem{conjecture}{Conjecture}
\usepackage{xr}

\newcommand\myshade{85}
\colorlet{mylinkcolor}{violet}
\colorlet{mycitecolor}{YellowOrange}
\colorlet{myurlcolor}{Aquamarine}

\hypersetup{
  linkcolor  = mylinkcolor!\myshade!black,
  citecolor  = mycitecolor!\myshade!black,
  urlcolor   = myurlcolor!\myshade!black,
  colorlinks = true,
}

\theoremstyle{plain}
\newtheorem{definition}{Definition}

\begin{document}
\title{A quantum generative model for multi-dimensional time series using Hamiltonian learning}

\author{Haim Horowitz}%
\author{Pooja Rao}
\author{\href{http://www.santoshkumarradha.me}{Santosh Kumar Radha}}
\email{santosh@agnostiq.ai}

\affiliation{%
 Agnostiq Inc. 325 Front St W, Toronto, ON M5V 2Y1\\
}%

\date{\today}

\begin{abstract}
Synthetic data generation has proven to be a promising solution for addressing data availability issues in various domains. Even more challenging is the generation of synthetic time series data, where one has to preserve temporal dynamics, \textit{i.e.}, the generated time series must respect the original relationships between variables across time. Recently proposed techniques such as generative adversarial networks (GANs) and quantum-GANs lack the ability to attend to the time series specific temporal correlations adequately. We propose using the inherent nature of quantum computers to simulate quantum dynamics as a technique to encode such features. We start by assuming that a given time series can be generated by a quantum process, after which we proceed to learn that quantum process using quantum machine learning. We then use the learned model to generate out-of-sample time series and show that it captures unique and complex features of the learned time series. We also study the class of time series that can be modeled using this technique. Finally, we experimentally demonstrate the proposed algorithm on an $11$-qubit trapped-ion quantum machine.
\end{abstract}

\maketitle

\paragraph{Introduction}: Extracting dynamical models and their qualitative characteristics from given time series data appears in various settings, including (but certainly not limited to) tracking~\cite{arulampalam2002tutorial}, medical imaging~\cite{penny2011statistical}, financial~\cite{zhang2004prediction}, and video analysis~\cite{dugad1999video}. Broadly, such methods can be classified into a spectrum of methods bounded by two extreme settings - deterministic dynamical models and generalized learning models. Deterministic dynamical models involve analytically modeling the process based on a specific set of assumptions. The model parameters are fit to agree with the data, which can then be used as a black box to analyze/predict the characteristics of the underlying process. In contrast, generalized learning models rely on searching and exploring for a solution in the hypothesis space. In general, a searchable hypothesis space can be pretty significant, even for reasonably simple data, and thus, one often restricts the search space by various techniques including, but not limited to, restricting the model architecture, learning procedure, etc. More concretely, we define the following problem - an unknown process emits multidimensional data at various time steps; with access limited to the emitted data only, can one recreate the process as a black box? In fact, many data sets, from language sentences to stock prices, assume that form. The simplest assumption to model such a process could be to start with a deterministic dynamical model, where we assume that the process emits data in each dimension sampled from a distribution, parameters of which can be fit for the given data. On the other extreme, one can train a deep recurrent neural network~\cite{NIPS2013_1ff8a7b5} to learn the process, training which would result in learning the black-box process. More often than not, it is unfavorable to over-parametrize a network in the hope of learning unknown features, as this increases the space of searchable models exponentially. An intriguing aspect of these learning methods is to explore architectures of networks that encode enough information about the problem to guide them to the solution model while being general enough to capture the required features.

\begin{figure}
\includegraphics[width=0.9\linewidth]{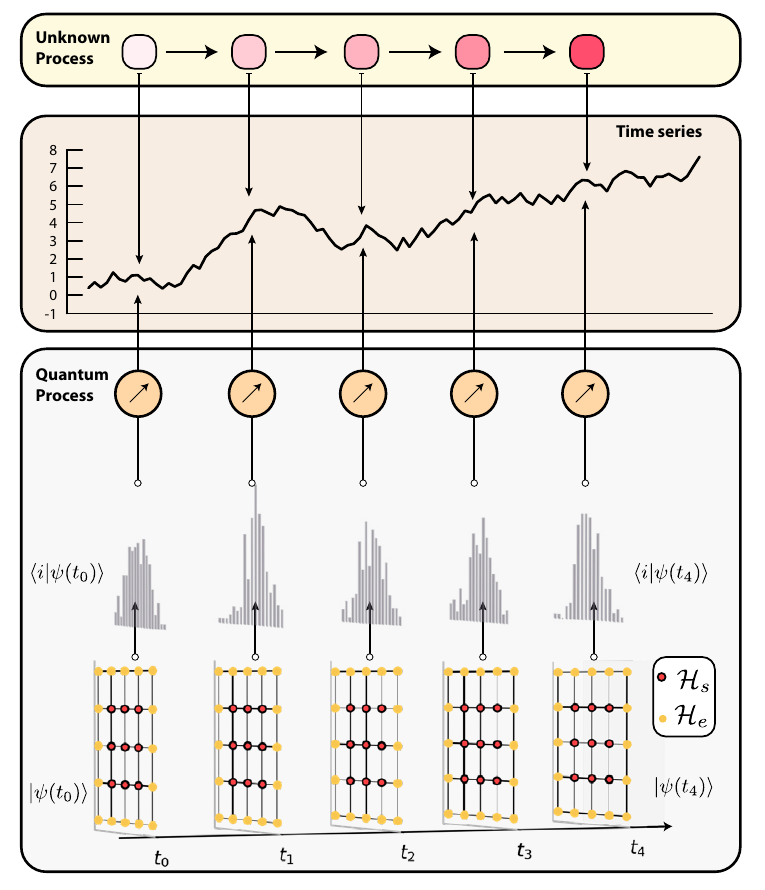}
    \caption{(Top) Unknown process emitting a signal at every time step. (Bottom) Open quantum system (red atoms) interacting with the external environment (yellow atoms) with time-dependent interactions being used as a learnt surrogate for the unknown process.}
    \label{fig:hero}
\end{figure}

\begin{figure*}[htb!]
\includegraphics[width=\linewidth]{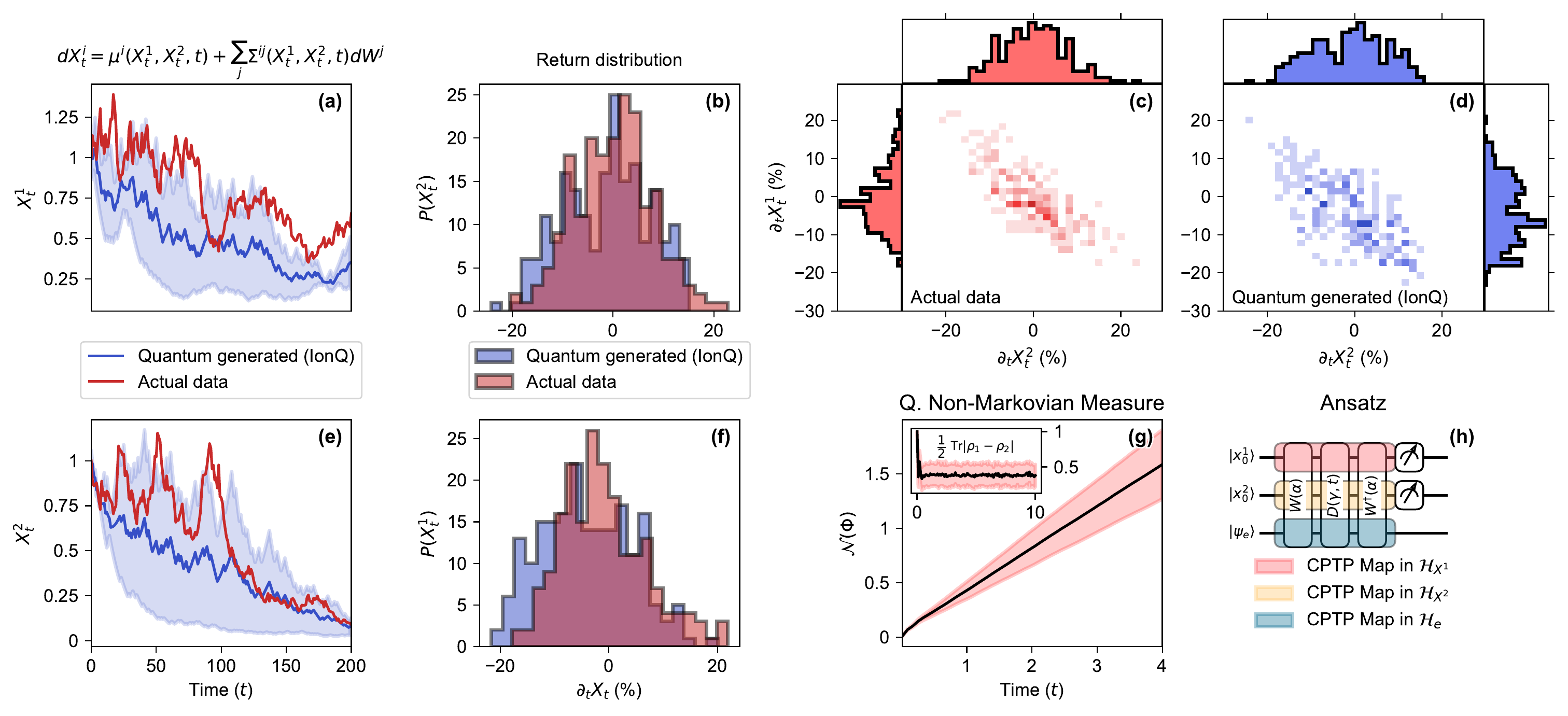}
    \caption{(a)/(e) Input 2-dimensional time series (red) and quantum generated time series' (blue) mean and variance. (b)/(f) First order percentage difference of the quantum generated (blue) and input time series (red). (c)/(d) Correlation between the 2-dimensional time series of the given data and mean of quantum generated series (Note the negative correlation feature being captured by the generated model). (g) Net quantum non-Markovian measure as given in \autoref{eq:nonmarkov} for a single dimensional time series. (h) Ansatz as defined in \autoref{eq:ansatz} along with the various CPTP channels for the two-dimensional time series.}
    \label{fig:mf}
\end{figure*}
Although there exist many networks specifically tuned for a particular model space, recently, physics-motivated architectures~\cite{daw2017, duan2022} have been explored to learn problems that are hard to generalize. Inspired by the these classical learning methods, we propose a physics-inspired quantum machine learning algorithm that can model the unknown process of a multi-dimensional time series. Intuitively, we propose using the dynamics of an open quantum system as a surrogate for the unknown process. We do this by embedding the open system ($\mathcal{H}_s$) in a bigger Hilbert space interacting with an environment($\mathcal{H}_e$). The goal then is to learn the Hamiltonian, which dictates the combined evolution of this entangled model living in $\mathcal{H}_s\otimes \mathcal{H}_e$. This model opens up the opportunity to dissect various information about the process. For instance, assessing changes in $\mathcal{H}_e$ at specific times could potentially give us insights into the process' environment. One could also, in principle, put constraints on the environment Hamiltonian based on the process' assumptions.

\paragraph{Formalism}: To better understand the class of time series being modeled via this technique, we develop the notion of $K$-Coherence. For a natural number $m$, a $m$-dimensional discrete time series is a sequence ${\mathbf{X}}=(\mathbf{s}_1,...,\mathbf{s}_m)$, all of which are of the same length, with possible states being $\{\ket{0},\ket{1}\ldots \ket{2^n-1}\}$. For $l\in \{1,...,m\}$, we define the transition probabilities for the time series $\mathbf{s}_l$, and denote $T_{ij}^l(k)$ as the probability of transitioning from the state $i$ to $j$ in $k$ steps for the time series $\mathbf{s}_l$. Let $K=\{k_0,...,k_{n-1}\}$ be a set of natural numbers. We say that that the high-dimensional time series ${\mathbf X}$ is $K-CPTP$ compatible if for every $k\in K$ and $l\in \{1,...,m\}$, there is a completely positive trace preserving ($CPTP$) map $\mathcal{E}_k^l: S(\mathcal{H}_{s_l}) \rightarrow S(\mathcal{H}_{s_l})$ such that 
\begin{equation}
\Tr\left(\dyad{j} \mathcal{E}_k^l(\dyad{i})\right)=T_{ij}^l(k) \label{eq:kcptp}.
\end{equation}
 As an immediate corollary to Stinespring's dilation theorem~\cite{SM}, we get the following \textemdash ~given a $K$-CPTP compatible time series ${\mathbf X}$, for every $k\in K$ and $l\in \{1,...,m\}$, there is a unitary $U_{k,l}$ witnessing the conclusion of Stinespring's theorem for the CPTP map $\mathcal{E}_k^l$. Therefore, given the $k$-step transition matrices of ${\mathbf X}$, it should be possible to learn the unitaries $U_{k,l}$'s for $k\in K$ and $l\in \{1,...,m\}$. Next, we will consider the case where our learning task is reduced to learning the matrices $U_{k,l}$ when they all arise from powers of the same matrix $U$. We start by defining the notion of $K$-Coherence in the following way.
\begin{definition} \label{k-coherence}
Let $K=\{k_0,...,k_{n-1}\}$ be a set of natural numbers. We say that an $m$-dimensional time series ${\mathbf X}=(\mathbf{s}_1,...,\mathbf{s}_m)$ is $K$-Coherent if:
\begin{enumerate}
 \item ${\mathbf X}$ is $K-CPTP$ compatible. 
 \item There is a finite-dimensional Hilbert space $\mathcal{H}_*$ such that, for every $l\in \{1,...,m\}$, $\mathcal{H}_{s_1} \otimes ... \otimes \mathcal{H}_{s_{l-1}} \otimes \mathcal{H}_{s_{l+1}} \otimes ... \otimes \mathcal{H}_{s_m} \otimes \mathcal{H}_*$ satisfies the conclusion of Stinespring's theorem for $\mathcal{E}_k^l$ for every $k\in K$.
 \item There is a unitary $U$ acting on $\mathcal{H}_{s_1} \otimes ... \otimes \mathcal{H}_{s_m} \otimes \mathcal{H}_*$ such that for every $l\in \{1,...,m\}$ and $k\in K$, the unitary $U_k=U_{k,l}$ from Stinespring's theorem for $\mathcal{E}_k^l$ satisfies $U_k=U^k$.
\end{enumerate}
\end{definition}
Here, $\mathcal{H}_{s_i}$ refers to the Hilbert space of the individual time series we are interested in, while $\mathcal{H}_{*}$ refers to the environment space that the given multidimensional time series is interacting with. We refer the reader to Ref.~\cite{SM} for a detailed discussion of  ${\mathbf X}=(\mathbf{s_1})$ one-dimensional time series. 

With Def.~\ref{k-coherence}, we set the formal requirement for time series that can be modeled via the powers of a single unitary matrix. It can be shown that the set of $K$-Coherent time series is strictly larger than those of time series formed from unistochastic transition matrices~\cite{SM}. While characterizing all time series that are $K$-Coherent is beyond the scope of this letter, we do numerically show below that a time series formed by a general stochastic differential equation comes under this set. We see that for any $K$-Coherent time series, one can always find a single parameter quantum unitary in a higher dimensional Hilbert space, the powers of which can model the temporal correlation of the original time series. As we will see in the subsequent  discussion, this amounts to finding a quantum Hamiltonian in a bigger latent space, the dynamics of which, has a unique map to the given series $\mathbf{X}$. Thus, one can, in theory, store the temporal information of such time series in a quantum channel (or equivalently a quantum Hamiltonian in higher dimensional space). Although the efficiency of such storage is to be evaluated, this gives us a technique to design a black box generative model that can be used to generate unseen time series with similar characteristics of the given series.

\paragraph{Algorithm}: Given ${\mathbf{X}}$, a classical $m$-dimensional $K$-Coherent time series, we start by considering $m$ Hilbert spaces, in which we embed the discrete classical states to their respective Hilbert space, the basis of which is spanned by the set of available classical states~\footnote{This is motivated by the fact that the quantum process modeling the classical process is a universal dynamical map, without which, one could in general map the initial state to a combined System+Environment Entangled state. }. In the previous discussion, we have seen that there exists a unitary $U$, the contents of which uniquely stores the higher-order transition probabilities of the given series, which can be accessed by its higher powers. Thus, in this formalism, we consider the existence of a systems$+$environment entangled single parameter unitary which guides the individual evolution of the time series. Equivalently using Stone's theorem, we are interested in the Hamiltonian $H$, that acts on the larger  $\mathcal{H}_{s_1} \otimes \ldots \otimes \mathcal{H}_*$ space, the time evolution of which results in the probability statistics of the given classical time series. To this end, similar to Ref.~\cite{qcl,C_rstoiu_2020} we parameterize the time dependent Hamiltonian (or the single parameter unitary) via the form 

\begin{equation}
U_{\vect{\alpha}}(t)=e^{-iH(\vect{\alpha})t}:=V_{lh}(\vect{\alpha}, t) := W_l(\theta)D_h(\gamma, t)W^{\dagger}_l(\theta), \label{eq:ansatz}
\end{equation}
where $W_l(\theta)$ is an $l$-layered parameterized quantum circuit, $D_h(\gamma, t)$ is an $h$-local parameterized diagonal unitary evolved for time $t$ and $\vect{\alpha} = \{\gamma, \theta\}$ is the parameter set~\cite{SM}. Trainability of this ansatz for general probabilities was previously shown in Ref.~\cite{qcl} and SM. We now define
\begin{equation*}
a_{ijkd}^{lh}(\vect{\alpha})=\Tr\left(\dyad{j}_{d} \Tr_{\bar{\mathcal{H}}_d}V_{lh}(\vect{\alpha}, k)\dyad{\bar{i}_d} V^\dagger_{lh}(\vect{\alpha}, k)\right).
\end{equation*}

Here, we use the notation $\ket{j}_{d}$ to indicate the $j^{th}$ basis state in Hilbert space $\mathcal{H}_{s_d}$, $\dyad{\bar{i}_d}$ to indicate a state given by $\dyad{0}_{s_1}\ldots\otimes\dyad{i}_{s_d}\otimes\ldots\dyad{0}_*$ and $\Tr_{\bar{\mathcal{H}}_d}$ is to trace all the Hilbert spaces but $\mathcal{H}_{s_d}$, the result of which is a mixed state $\rho_d\in S(\mathcal{H}_{s_d})$. More verbally, we start by initializing state $\ket{i}\in \mathcal{H}_{s_d}$, after which we evolve the state in the bigger combined Hilbert space for time $k$, using the parameterized unitary of the form $V_{lh}(\vect{\alpha}, t)$ with $l$-layered $W$ and $h$-layered $D$ ansatz. We then trace out the effect of the pseudo environment, \textit{i.e}, $\mathcal{H}_*$ + systems other than $\mathcal{H}_{s_d}$, resulting in $\rho_d\in S(\mathcal{H}_{s_d})$. We finally project it on to state $\ket{j}_d\in\mathcal{H}_{s_d}$. Thus $a_{ijkd}^{lh}(\vect{\alpha})$ gives us the probability measure of how likely it is for state $\ket{i}_d\rightarrow\ket{j}_d$ using a $CPTP$ map when evolved for time $k$ using a parameterized circuit $V_{lh}(\vect{\alpha}, t)$.

Given a time series ${\mathbf{X}}$, it is straight forward to calculate the transition matrix for the $d^{th}$ time series for $k^{th}$ time step $T_{ij}^d(k)$. Thus one can define the cost function as 
\begin{equation*} 
C_{lh}(\vect{\alpha})=\sum_{ijdk}\mathcal{D}\left(a_{ijkd}^{lh}(\vect{\alpha}),T_{ij}^d(k)\right),\label{eq:cost}\\
\end{equation*}

and the corresponding objective function as 
\begin{equation*} 
\argminG_{\vect{\alpha}} C_{lh}(\vect{\alpha})\label{eq:loss},
\end{equation*}

where $\mathcal{D}(\cdot,\cdot)$ is an appropriate measure of loss. It is important to note that the vector formed by suppressing the index $j$ for both $a$ and $T$ gives us a probability vector that sums up to one. This probability vector for $a$ can be efficiently calculated in a quantum computer as the final statistical measurement on the qubits in the given $\mathcal{H}_{s_d}$ space. We thus use the Kullback-Leibler divergence as the loss measure in this index and calculate the following as the loss function throughout the letter,
\begin{equation} 
C_{lh}(\vect{\alpha})=\sum_{idk}\left(\sum_j a_{ijkd}^{lh}(\vect{\alpha}) \log \left(\frac{a_{ijkd}^{lh}(\vect{\alpha})}{T_{ij}^d(k)}\right) \right)\label{eq:actualloss}.
\end{equation}

Complexity of this sum is given by the number of states in each Hilbert space times the total time steps. For a $K$ time step learning with each $d$ time series spanning a space of $n_1,n_2\cdots n_d$ qubits, we have $O(\sum_{i=1}^{d}2^{2n_i}K)$ elements going into the summation. Typically, one randomizes the elements in this loss function into smaller mini-batch of size $<<\sum_{i=1}^{d}2^{2n_i}K$.

\paragraph{Experimental Demonstration}: Because of the fairly adaptable nature of resource requirements, the techniques introduced above are directly applicable to current day NISQ hardware. As a concrete example, we start by first generating a synthetic pair of correlated time series using the stochastic differential equation $dX^i_t = \mu^i dt + \sum_{j} \sigma^{ij}dW_t^j$, where $W_t$ represents the Brownian motion, with a drift term given by $\mu=(-1,-2)$ and correlation term given by $\sigma=\bigl( \begin{smallmatrix}1 & -0.5\\ -0.5 & 1\end{smallmatrix}\bigr)$. This series is shown in \autoref{fig:mf}(a/e) as the red curve. Given the continuous and unbounded nature of the time series, we first make the series stationary by converting it into first order difference series. To convert this continuous stationary series into a discrete one, we bin the data of both series into $2^{n_{s_i}}=2$ states by using Symbolic Aggregate approXimation (SAX)~\cite{lin2007experiencing} transformation. Thus, we can model each of the systems using only one qubit, thus requiring two qubits in total for the system. Though there is a theoretical upper bound for the environmental dimensionality from Stinespring dilation theorem, due to the inability of ansatz to access the entire Hilbert space, this is only a weak upper bound. We thus choose $n_e=2$ qubits as the environment dimension, leading to 5 qubits in total for the combined system that includes the systems and its environment. One can compute the $2\times2$ transition matrix $T$ for time steps $\mathcal{T}_k=\{1,2,10,30,50\}$. With this, we choose the Strongly Entangling circuit ansatz~\cite{2020ansatz,SM} for $W$ with $l=2$ and choose $h=2$ connected diagonal matrix as our ansatz. We start by initializing the respective subspace with the basis vector to generate the $\ket{i}^{th}$ basis after which we evolve it for said time steps $\mathcal{T}_k$ to compute the cost function given in \autoref{eq:actualloss}. We choose a mini-batch size of $30$ for each iteration and use the gradient-free optimization routine \lstinline{COBYLA}~\cite{powell1994direct}. Due to various practical considerations, we run the optimization routine to find the optimal $\vect{\alpha}^*$ using classical simulations in  \lstinline{PennyLane}~\cite{bergholm2018pennylane} after which we turn to the $11$-qubit trapped ion quantum hardware provided by \lstinline{IonQ}~\cite{2019ionq} for generating the final results using the optimal parameters shown in~\autoref{fig:mf}. 

Once the model is trained, we use the model by initializing with a starting state ($\ket{0}\otimes\ket{0}$ in our case), and based on the requirement, one can proceed with multiple read out techniques. Here, since we are interested in generating similar time series, we do a single shot readout for various time steps and collect the statistics. \autoref{fig:mf}(a/e) (blue) show the mean and variance of cumulative time series obtained by the generated by the two-state model using $\vect{\alpha}^*$ and (b/f) shows the corresponding distribution of the first-order difference of the mean series. For convenience, we refer to the mean of the generated series as $\mathbf{X}_q$ and the training data as $\mathbf{X}_d$. First, we see that the temporal feature of the time series, \textit{i.e.}, the exponentially decreasing nature of the synthetic data has been recognized by the generative model. Second, by comparing the properties of first order differences ($\partial_t\mathbf{X}_q,\partial_t\mathbf{X}_d$), we see that the model is capable of accounting for higher order temporal features in the data. This ability can be attributed to the learning process, during which we encode temporal features. First order differences play a significant role in real world applications like financial time series, where it is a stationary return series.  We now turn our focus to \autoref{fig:mf}(c/d) where we show the joint probability distribution of $\partial_t\mathbf{X}_q$ and $\partial_t\mathbf{X}_d$, where in (c) we see the apparent negative correlation we have encoded between $\mathbf{X}_d$. \autoref{fig:mf}(d) shows the correlation between generated $\mathbf{X}_q$. Our result demonstrates that the model is not just capable of capturing the effect of negative correlation but also, to a reasonable extent, the degree of correlations.

\paragraph{Discussion}: Quantum extension (qGANs) of classical Generative Networks have recently been proposed~\cite{demers2018,2019qgan}. In Ref.~\cite{zhu2021generative}, authors use qGANs to solve problems of similar spirit, however, we see that our model has the ability to capture temporal features better than that of GANs. This can be attributed to the specificity of our networks. Networks such as the one proposed in this letter are tuned specifically for learning temporal data where we associate a one-to-one map between the time evolution of our data and the quantum process we are learning. This is in stark contrast to general-purpose networks like qGANs (or Deep Neural Networks), where one gives the network the freedom to learn all the features while expecting (hoping) it to learn the minimal set required to capture the needed details.  This amounts to finding the correct model from a more expansive search space of optimal network weights than architectured networks.  Having said that, general networks can be applied to a broader set of problems while the current network is better suited for dynamical problems.

We now discuss the potential quantum advantage of using the proposed network. We start by studying the class of the learned quantum dynamical process. One can measure the degree of quantum non-Markovianity~\cite{breuer2009,SM} associated with a quantum process ($\Phi$) as 
\begin{equation}
\mathcal{N}(\Phi)=\max _{\rho_{1,2}(0)} \int_{\sigma>0} d t \sigma\left(t, \rho_{1,2}(0)\right),\label{eq:nonmarkov}
\end{equation}
where $\sigma\left(t, \rho_{1,2}(0)\right)=\frac{d}{d t} D\left(\rho_{1}(t), \rho_{2}(t)\right)$ depends on the time $t$, initial states $\rho_{1,2}(0)$ with their time evolutions given by $\rho_{1,2}(t)=\Phi(t, 0) \rho_{1,2}(0)$ and $D(\cdot,\cdot)$ is the trace distance. Although a complete numerical measure of $\mathcal{N}(\Phi)$ requires a large sampling of initial $\rho_{1,2}(0)$, our goal here is to show that the quantum process that a given classical series maps to is indeed that of a quantum non-Markovian open system. In \autoref{fig:mf}(g), we plot $\mathcal{N}(\Phi)$ for a 1D time series modeled using a $1$-qubit system and a $2$-qubit environment with $\Phi$ being the quantum CPTP process that evolves the system state. Inset shows that the distance $D\left(\rho_{1}(t), \rho_{2}(t)\right)$ between a pair of initial states as a function of time. We clearly see that process is not only non-Markovian but the strength of non-Markovianity increases with time. From the result, it is evident that optimal quantum simulation of such classical time series involves quantum non-Markov processes, which in general, are inefficient to be accessed classically and can only be classically simulated under certain approximations~\cite{Luchnikov2019}.

Though a basic framework of $K$-Coherence was introduced in the previous section, further work is needed to understand the class of classical time series processes that belong to this set. It can be shown that there exists time series that are $K$-Coherent for some $K$, but cannot be modeled by a closed system. However, it is not known whether there is a time series that cannot be $K$-Coherent for any $K$~\cite{SM}. Recent works \cite{2020markov,2018time} have explored similar questions relating to quantum time series and their classical nature; however, we are concerned with the inverse problem --- when does a classical time series correspond to that of an open quantum system? Note that in the current setup, even though we do not allow for the most general form of open evolution, one can easily reformulate the cost function to include higher-order quantum maps --- comb processes~\cite{2008combs1,Chiri2009}. This can be done by modifying the ansatz from $V_{lh}(\vect{\alpha}, t)\rightarrow\Pi_i V^i_{lh}(\vect{\alpha}_i, t)$, where the quantum channel of each unitary acts as individual independent CPTP maps. Similar works~\cite{Agha2018, Agha2016, Korzekwa2021} have shown advantage when it comes to simulating a classical process using a quantum process for simpler assumptions on classical process. It is yet to be seen if such properties exist for $K$-Coherent processes. It is interesting to note that the problem of finding such unitary turns out to be closely related to the problem of solving a set of polynomials with integer coefficients, which is known to be hard~\cite{SM}.

In summary, we have developed a framework that utilizes a quantum resource as a black box to learn the features of an unknown process. We then propose using this quantum process as a generative network to generate unseen samples with characteristics similar to the unknown process. Our numerical and hardware results also demonstrate the model's capability in learning higher-order features. With this method, we hope to provide a quantum network with the minimum structure required to learn temporal features without losing its ability to generalize. We show potential quantum advantages in two ways. First, by showing that the learning and simulation of certain classical processes can be quantum non-Markovian in nature. Second, by providing evidence that this particular problem of learning involves solving systems of polynomial equations in a particular way, which is computationally hard in general. Nevertheless, there are still critical questions that need to be addressed in future works, such as the stability of the network with respect to noise in quantum hardware, the size of the set of $K$-Coherent time series, bounds on environment resource requirement arising due to incapability of unitary ansatz to explore the entire space.

\paragraph{Acknowledgement:} We acknowledge the use of IonQ Quantum Hardware and associated credits for this work. The views expressed are those of the authors, and do not reflect the official policy or position of IonQ team. We also acknowledge Jack S. Baker, Will Cunningham and Oktay Goktas for their discussion on optimization and Markov process and general project guide respectively.

\begin{figure}[H]
\centering
\includegraphics[width=.5\linewidth]{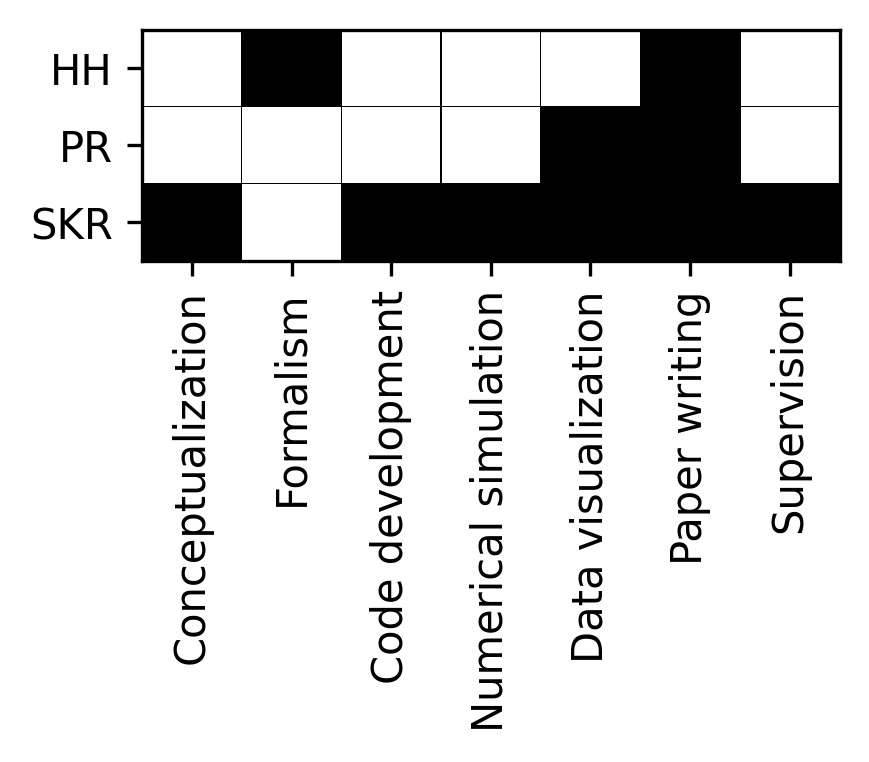}
    \caption{Contribution of various authors. Black box indicates the areas various authors contributed to. }
    \label{fig:contri}
\end{figure}

\bibliography{refs}

\clearpage
\newpage
\mbox{~}
\onecolumngrid
\begin{center}
\textbf{\large Supplemental Materials: A quantum generative model for multi-dimensional time series using Hamiltonian learning}
\end{center}

\setcounter{equation}{0}
\setcounter{figure}{0}
\setcounter{table}{0}
\setcounter{page}{1}
\makeatletter
\renewcommand{\theequation}{S\arabic{equation}}
\renewcommand{\thefigure}{S\arabic{figure}}
\renewcommand{\bibnumfmt}[1]{[S#1]}
\renewcommand{\citenumfont}[1]{S#1}

\tableofcontents

\section{Introduction}

We provide more details to support the main text. Organization of supplemental material is as follows. We start with a general quantum process~\autoref{ssec:quantum process}, which is dynamically evolved under both closed (\autoref{ssec:closed}) and open (\autoref{ssec:open}) quantum channels. It is followed by a discussion on the notion of $K-$Coherence in \autoref{sm:1d_timeseries}, where we illustrate our analysis using a $2$-qubit example. We then proceed to discuss the hardness of the learning problem in \autoref{ssec:hardness}. In \autoref{ssec:canonicald}, we demonstrate numerically the existence of a unique Hamiltonian that models the time series. A complete description of the ansatz as well as the learning procedure used in this study is given in \autoref{ssec:learn}. Finally in \autoref{ssec:nonmarkov}, we numerically explore the ability of our model to simulate non-Markovian dynamics.

\section{Quantum process}\label{ssec:quantum process}
We begin with the notion of a (pure) quantum state $\ket{\psi}$, which is a vector in a $2^n$-dimensional Hilbert space $\mathcal{H}_s$, spanned by mapping the classical states (numbered $0,1,...,2^n-1$) to the quantum states $\ket{0}, \ket{1},...,\ket{2^n-1}$ in $\mathcal{H}_s$ by $i \rightarrow \ket{i}$. Any such pure state can be written in terms of its basis given by
\begin{equation}
    \ket{\psi}=\sum_i c_i\ket{i},\label{eq:1}
\end{equation}
where $c_i$ could be a complex number, but here we restrict it to be a real number, that is, $c_i=\sqrt{p_i}$ for some $p_i \geq 0$. The quantum state is a complex vector in a Hilbert space, which is required to satisfy the property $\expval{\psi | \psi}=1$, implying $\sum_i p_i =1$.

\subsubsection{Closed quantum modeling}\label{ssec:closed}

We fix a quantum Hamiltonian ${H}$ and time evolve with initial state $\ket{\psi(0)}$ given by
\begin{align}
\ket{\psi(t)}=e^{-i{H}t}\ket{\psi(0)}:=U^t\ket{\psi(0)}\label{eq:closed time evolution},
\end{align}
where $U=e^{-i{H}}$. For $i<2^n$ and natural $t$, we let $\ket{\psi_i (t)}=e^{-i {H} t} \ket{i}$, so $\ket{\psi_i (t)}=\ket{\psi (t)}$ when $\ket{\psi (0)}=\ket{i}$. For a time series $\bold X$, natural $t$ and $i, j<2^n$, we let $p_j^i(t)$ be the probability of measuring $\ket{j}$ starting from $\ket{\psi(0)}= \ket{i}$ after $t$ steps. So, assuming that the time series was obtained from a Hamiltonian as above, we have $\sqrt{p_j^i(t)}=\bra{j} \ket{\psi_i(t)}$. It follows that $\ket{\psi_i (t)}=\sum_{j<2^n} \sqrt{p_j^i (t)}\ket{j}$.
 This corresponds to the closed evolution of quantum state  $\ket{\psi(0)}$ under ${H}$.

Any unitary matrix $U$ of dimension $N$ can naturally be linked to a stochastic transition matrix $T$ by defining the $(i, j)$-entry of $T$ as follows:
\begin{equation*}
    U_{ij} = m_{ij}e^{i\phi_{ij}}\rightarrow T_{ij}=|U_{ij}|^2=m^2_{ij} \label{eq:a1}
\end{equation*}
where $m_{ij},\phi_{ij} \in \mathrm{R}$. Because of the unitarity of $U$, it is clear that the matrix $T$ is stochastic, \textit{i.e.},
\begin{align*}
    T_{ij} &\geq 0 \\
  \sum_{j} T_{ij} &= 1
\end{align*}

However, not every stochastic matrix $T$ corresponds to a unitary $U$. The set of matrices that correspond to unitary matrices are called doubly stochastic  matrix~\cite{marshall1979inequalities} and have the additional property $\sum_i T_{ij}=1$ $\forall i$. A doubly stochastic matrix is, by definition, a matrix of non-negative elements such that the elements in each row and each column sum up to 1. And not every doubly-stochastic matrix can be associated with a unitary matrix as the rows and columns of a unitary matrix have to obey orthogonality conditions which imposes further restrictions on the matrix elements. One therefore defines the subset of doubly stochastic matrices $T$ which satisfy $T_{ij} = \norm{U_{ij}}^2$ for some unitary matrix $U$ as unistochastic  matrices~\cite{marshall1979inequalities}.

Suppose that we are given a time series $\mathbf{s}$ and a sequence $K=\{k_0,...,k_{n-1}\}$ of time steps. Suppose in addition that for each $k\in K$, the matrix $T(k)$ whose $(i, j)$-entry is the probability of going from $i$ to $j$ in $k$ steps, is a  unistochastic matrix $T(k)$. The goal is to find a fixed quantum Hamiltonian ${H_s}$, whose unitary evolution ($U_s^k=e^{-i{H_s}k}$) satisfies the following
\begin{equation*}
    \norm{\mel{j}{U_s^k}{i}}^2=p_j^i(k)=T_{ij}(k)\quad \forall k\in K.
\end{equation*}

Here, we implicitly assume the existence of such unitaries for our time series. Note that unistochasticity means $\sum_ip_j^i(k)=\sum_jp_j^i(k)=1$.

It is also important to note that the problem of finding such a Hamiltonian is closely related to finding a unitary matrix which when raised to a general power $k$ with individual elements squared, results in matrix $T(k)$ for every $k \in K$.

\subsubsection{Open quantum modeling}\label{ssec:open}

We start by considering the case where the system defined in \autoref{eq:1} is part of a bigger Hilbert space $\mathcal{H}=\mathcal{H}_s\otimes \mathcal{H}_e$ with $\mathcal{H}_e$ corresponding to some outer environment. Let $\bm{\rho}$ be the total density matrix of this combined system. Now extending the definition in \autoref{eq:1}, we have $\rho_s(t)=\ket{\psi_s(t)}\bra{\psi_s(t)}$. In this setting, given a total density matrix $\bm{\rho}(0)$ at time $t=0$, final state of the system is given by the partial trace
\begin{equation*}
    \rho_{s}\left(t\right)=\operatorname{Tr}_{e}\left[U\left(t\right) \bm{\rho}\left(0\right) U^{\dagger}\left(t\right)\right] \label{eq:open time ev}.
\end{equation*}

In general, total evolution is not factorizable, \textit{i.e.}, $U\neq U_s \otimes U_e$ (where $U_s$ and $U_e$ are obtained by restricting $U$ to the basis elements of the system and the environment, respectively), and thus both the system and the environment are modeled to be interacting with each other. One can now analyze \autoref{eq:closed time evolution} and rewrite it as a \textit{dynamical map} acting purely on the system $\mathcal{H}_s$ and connecting the states of the system $S$ at time $0$ to $t$, given by 
\begin{equation*}
    \mathcal{E}_{\left(t, 0\right)}: \rho_{s}\left(0\right) \rightarrow \rho_{s}\left(t\right). \label{arrows}
\end{equation*}

In general, $\mathcal{E}$ not only depends on the global evolution operator $U$, but also on the properties of system and environment. We will now restrict ourselves to a special class of maps called the Universal Dynamical Maps (UDM), which is defined as follows~\cite{bonzio2013open}.

\begin{definition}
A dynamical map is a UDM if and only if it is induced from an extended system with the initial condition $\bm{\rho}(0)=\rho_s(0)\otimes \rho_e(0)$, where $\rho_e(0)$ is fixed for $\rho_s(0)$. That is, a dynamical map is a UDM if the corresponding arrow makes the diagram in Figure 1 commute.
\end{definition}

These maps require the system to start from separable states, and the entanglement is introduced subsequently. This can be made clear with the representation shown in \autoref{fig:maps}.

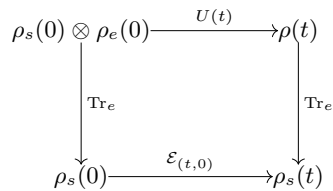
\begin{figure}
    \centering
\begin{tikzcd}[row sep=5em, column sep=5em] %
    \rho_s(0)\otimes \rho_e(0)
    \arrow{r}{U(t)} 
    \arrow{d}{\text{Tr}_e}
&
    \bm{\rho}(t)
    \arrow{d}{\text{Tr}_e}   \\
    \rho_s(0)
    \arrow{r}{\mathcal{E}_{\left(t, 0\right)}}
&
    \rho_s(t)
\end{tikzcd}
\caption{A universal dynamical map being considered  and its relation to the evolution of system $+$ environment.}\label{fig:maps}
\end{figure}

This consideration of UDMs instead of generic maps is important because we want the system to be separable at initial time. However, this does not restrict the system from becoming entangled as it evolves. For instance, without restricting oneself to UDMs, one could very well map the discrete classical states to a state that is entangled with the environment --- $i \rightarrow \rho^i_{s+e}$, instead we only consider classical to quantum mappings which are given by the form $i \rightarrow \ket{i}\otimes\ket{e}$ (more specifically we set $\ket{e}=\ket{0}^{\otimes n_e}$, where $n_e$ is the size of environment Hilbert space).

With this setting, given a classical time series $X\in S^n$ with transition matrix $T(k)$ for set of $K:=\{k_0,k_1..\}$, the goal in the open case is to find a global Hamiltonian ${H}\in \mathcal{H}$, the unitary evolution ($U^k=e^{-i{H}k}$) of which satisfies
\begin{equation*}
  \operatorname{Tr}_{s} \operatorname{\hat{P}_j} \operatorname{Tr}_{e}\left[ U\left(k\right) \rho_i \left(0\right) U^{\dagger}\left(k\right)\right]=T_{ij}(k)\quad \forall k\in K, 
\end{equation*}

where $\hat{P}_j$ is the projection operator on the $j^{th}$ basis state of system and $\rho_i=\ket{i}\bra{i}_s\otimes\ket{0}\bra{0}_e$. Unlike the previously considered closed case, we are not constrained by the unitarity of $U$ as $\mathcal{E}$ can be any general universal dynamical map. This lifts off the requirement that $T(k)$ needs to be unistochastic and thus, it can model a more general transition matrix. \autoref{fig:maps} gives the visual analogue of this process. 

To better understand the resource needed for this type of system, we can invoke Stinespring's theorem. Even though the theorem applies more generally to completely positive (not necessarily trace-preserving) maps between $C^*-$algebras, we use it for our UDMs. 

\begin{theorem}[Stinespring's dilation] \label{stinespring}
Let $\mathcal{E}:S(\mathcal{H}_s) \mapsto S(\mathcal{H}_s)$ be a completely positive and trace-preserving map between states on a finite-dimensional Hilbert space $\mathcal{H}_s$ of dimension $n_s$. Then there exists a Hilbert space $\mathcal{H}_e$ with dimension $n_e$ and a unitary operation $U$ on $\mathcal{H}_s \otimes \mathcal{H}_e$  such that
\begin{equation*}
    \mathcal{E}(\rho_s)=\text{Tr}_e\left[U \left(\rho_s(0)  \otimes \rho_e(0)\right)U^{\dagger}\right]\label{eq:stin}
\end{equation*}
$\forall \rho \in S(\mathcal{H}_s) $ with $n_e \leq n_s^2$.
\end{theorem}

Thus, in the simplest case where the time series is modeled as an open quantum system interacting with environment, the maximum dimensionality needed to describe the environment scales quadratically. 

In both open and closed cases, the goal turns out to learn the time dependent unitary $U(t)=e^{-i{H}t}$, which in turn results in learning the Hamiltonian ${H}$ that is responsible for producing the quantum states that have the same quantum correlations \textit{entanglement amplitudes} as the the classical correlations of classical time series. In the open case, we not only learn the Hamiltonian of the system(${H}_s$), but also the Hamiltonian of environment (${H}_e$) and the coupling between them (${H}_{se}$) as given in \autoref{eq:ham}.

\begin{equation*}
{H}=
\begin{bmatrix}
\,H_s & H_{se}\,  \\
\,H_{se} & H_e\, \\
\end{bmatrix}. \label{eq:ham}
\end{equation*}
 
It is easy to see that the closed case is a special case where ${H}_e={H}_{se}=0$. This formulation gives us potential insights on the classical system's environment/bath as well as its interaction by looking at the Hamiltonian terms in respective subspaces.

\section{On the extent of $K$-Coherence} \label{sm:1d_timeseries}
In this section, we illustrate the theoretical ideas behind our algorithm using the case of a one-dimensional time series. While these ideas can be formulated for arbitrarily large systems, we shall consider a fixed $1$-qubit system $\mathcal H$ and a fixed $1$-qubit environment $\mathcal{H}_*$. We shall also assume for simplicity that all unitaries in question have real entries. The general case is similar though somewhat more technically complicated. As we saw earlier, our time series is obtained from some underlying unitary $U$ in the sense that its transition probabilities satisfy the equation $T_{ij}(k)=Tr(\ket{j}\bra{j} Tr_{\mathcal{H}_*}(U^k \ket{i} \bra{i} \otimes \ket{0} \bra{0} {U^{k\dagger}}) )$ for all relevant $i, j$ and $k$. Our task was to learn a unitary that will give rise to similar transition probabilities. We shall show now that this is equivalent to finding a unitary whose entries solve a certain set of polynomial equations.
Let 
\begin{align*}
\rho_0 &= \ket{0} \bra{0} \otimes \ket{0} \bra{0}= \begin{pNiceMatrix} 
    1 & 0 & 0 & 0\\ 
    0 & 0 & 0 & 0\\ 
    0 & 0 & 0 & 0\\ 
    0 & 0 & 0 & 0 
\end{pNiceMatrix},\\
\rho_1 &= \ket{1} \bra{1} \otimes \ket{0} \bra{0}= \begin{pNiceMatrix} 
    0&0 &0 &0 \\ 
    0&0&0&0 \\ 
    0&0&1&0 \\ 
    0&0&0&0 
\end{pNiceMatrix}.
\end{align*}
Let $U$ be a $4 \times 4$ unitary and fix a natural number $k$. Suppose that $U^k=(a_{ij})$, then ${U^k}^{\dagger}=(a_{ji})$. One can now show that 
\begin{align*}
U^k \rho_0 {U^k}^{\dagger}= \begin{pNiceMatrix} a_{11}a_{11} & a_{11}a_{21} & a_{11}a_{31} & a_{11}a_{41}\\ a_{21}a_{11} & a_{21} a_{21} & a_{21}a_{31} & a_{21}a_{41} \\ a_{31}a_{11} & a_{31} a_{21} & a_{31}a_{31} & a_{31}a_{41} \\
a_{41} a_{11} & a_{41}a_{21} & a_{41} a_{31} & a_{41} a_{41}
\end{pNiceMatrix}
\end{align*}
along with 
\begin{equation*}
U^k \rho_1 U^{k\dagger} = 
\begin{pNiceMatrix} 
a_{13}a_{13} & a_{13}a_{23} & a_{13}a_{33} & a_{13}a_{43}\\ a_{23}a_{13} & a_{23} a_{23} & a_{23}a_{33} & a_{23}a_{43} \\ a_{33}a_{13} & a_{33} a_{23} & a_{33}a_{33} & a_{33}a_{43} \\
a_{43} a_{13} & a_{43}a_{23} & a_{43} a_{33} & a_{43} a_{43}
\end{pNiceMatrix}. 
\end{equation*}
Let $i\in \{0, 1\}$ and suppose that 
\begin{align*}
Tr_{\mathcal{H}_*}(U^k \rho_i {U^k}^{\dagger})= \begin{pNiceMatrix} a & b \\ c & d \end{pNiceMatrix},
\end{align*}
then
\begin{align*}
\ket{0} \bra{0}Tr_{\mathcal{H}_*}(U^k \rho_i {U^k}^{\dagger})=\begin{pNiceMatrix} a & b \\ 0 & 0 \end{pNiceMatrix},
\end{align*}
hence $Tr\left(\ket{0} \bra{0}Tr_{\mathcal{H}_*}(U^k \rho_i {U^k}^{\dagger})\right)=a$. Similarly, 
\begin{align*}
Tr\left(|1)(1| Tr_{\mathcal{H}_*}(U^k \rho_i {U^k}^{\dagger})\right)=Tr\begin{pNiceMatrix} 0 & 0 \\ c & d \end{pNiceMatrix}=d.
\end{align*}
Now consider the basis $\{\ket{0} \otimes \ket{0}, \ket{0} \otimes \ket{1}, \ket{1} \otimes \ket{0}, \ket{1} \otimes \ket{1}\}$ for $\mathcal H \otimes \mathcal{H}_*$, then by the formula for computing the partial trace we get, for $i=0$:
\begin{align*}
    a &= a_{11}a_{11}+a_{21}a_{21}\\
    d &= a_{31}a_{31}+a_{41}a_{41},
\end{align*}
and for $i=1$:
\begin{align}
a &= a_{13}a_{13}+a_{23}a_{23} \\
d &= a_{33}a_{33}+a_{43}a_{43}
\end{align}
\\
We therefore arrive at the following conclusion.

\textbf{Conclusion}: Let $U$ be a $4 \times 4$ unitary and let $k$ be a natural number and denote $U^k=(a_{ij})$. For $i, j \in \{0, 1\}$, let $T_{ij}'(k):=Tr(\ket{j} \bra{j} Tr_{\mathcal{H}_*}(U^k \ket{i} \bra{i} \otimes \ket{0} \bra{0} {U^{k\dagger}}))$,  then we have,

\begin{align*}
T_{00}'(k) &= a_{11}a_{11}+a_{21}a_{21}\\
T_{01}'(k) &= a_{31}a_{31}+a_{41}a_{41}\\
T_{10}'(k) &= a_{13}a_{13}+a_{23}a_{23}\\
T_{11}'(k) &= a_{33}a_{33}+a_{43}a_{43}
\end{align*}

Now, let us continue using the notation in the above conclusion and let $(b_1,...,b_{16})$ enumerate the entries of $U$, then for every $i, j \in \{1,...,4 \}$, there is a polynomial $p_{ij}^k \in Z[x_1,...,x_{16}]$ such that $a_{ij}=p_{ij}^k(b_1,...,b_{16})$ (note that the polynomials do not depend on $U$ but only on $i, j$ and $k$). Therefore, the conclusion can be reformulated using the following polynomial equations:

\begin{align}
T_{00}'(k) &= (p_{11}^k p_{11}^k +p_{21}^k p_{21}^k)(b_1,...,b_{16})\\
T_{01}'(k) &= (p_{31}^k p_{31}^k+ p_{41}^k p_{41}^k)(b_1,...,b_{16})\\
T_{10}'(k) &= (p_{13}^k p_{13}^k+p_{23}^k p_{23}^k)(b_1,...,b_{16})\\
T_{11}'(k) &= (p_{33}^k p_{33}^k+p_{43}^k p_{43}^k)(b_1,...,b_{16}).
\label{eqn:poly}
\end{align}
Therefore, given a set of transition probabilities $\{T_{ij}(k) : i, j \in \{0, 1\}, k\in K\}$, finding a unitary that satisfies  $T_{ij}(k)=Tr(\ket{j} \bra{j} Tr_{\mathcal{H}_*}(U^k \ket{i} \bra{i} \otimes \ket{0} \bra{0} {(U^{\dagger})^k))}$ for all $i, j$ and $k$ is equivalent to finding a unitary $U=(b_{ij})$ whose entries solve the polynomial equations~\autoref{eqn:poly1} above for all $k\in K$ where $T_{ij}'(k)$ on the left hand side is replaced by $T_{ij}(k)$. Recall that the assumption of $K$-Coherence implies that there is at least one unitary as required. We shall discuss later the implications of solving systems of polynomial equations using unitary matrices.

 We shall now formulate and prove the observation that the transition probabilities that are induced from unitaries on a system and environment via CPTP maps are exactly the transition probabilities that satisfy the definition of $K$-Coherence. 

\begin{theorem} If $U$ is a unitary as before, then $\mathcal{E}(\rho)=Tr_{\mathcal{H}_*}(U \rho \otimes \ket{0} \bra{0} U^{\dagger})$ is CPTP.
\end{theorem}

\begin{proof}
Given $\rho= \begin{pNiceMatrix} a & b \\ c & d \end{pNiceMatrix}$, we have 
\begin{align}
\rho \otimes \ket{0} \bra{0}= \begin{pNiceMatrix} a & 0 & b & 0 \\ 0 & 0& 0 & 0 \\ c & 0 & d & 0 \\ 0 & 0 & 0 & 0 \end{pNiceMatrix},
\end{align}
so $Tr(\rho)=Tr(\rho \otimes \ket{0} \bra{0})$, hence $\rho \mapsto \rho \otimes \ket{0} \bra{0}$ is trace preserving. A similar computation shows that it's completely positive, hence CPTP. The maps $A \mapsto UAU^{\dagger}$ and $Tr_{\mathcal{H}_*}$ are known to be CPTP. As CPTP maps are closed under composition, we are done. 
\end{proof}

\begin{definition}
Given a unitary $U$ acting on $\mathcal H \otimes \mathcal{H}_*$ and a natural number $k$, let $\mathcal{E}_k^U$ be the CPTP map $\mathcal{E}_k^U (\rho)=Tr_{\mathcal{H}_*}(U^k \rho \otimes \ket{0} \bra{0}{U^{k\dagger}})$. By the precious theorem, $\mathcal{E}_k^U$ is indeed CPTP.
\end{definition}

We shall now prove the existence of a $K$-coherent time series that cannot be modelled by a closed system.

\begin{theorem} \label{1d-series}
There exists a $\{1\}$-Coherent time series that cannot be modelled by a closed system.
\end{theorem}

\begin{proof}
Let $U=CNOT(H \otimes I)=\frac{1}{2} \begin{pNiceMatrix} 1 & \phantom{-}0 & \phantom{-}1 & \phantom{-}0 \\ 0 & \phantom{-}1 & \phantom{-}0 &-1 \\ 1 & \phantom{-}0 & -1 & \phantom{-}0 \\ 0 & \phantom{-}1 & \phantom{-}0 & \phantom{-}1 \end{pNiceMatrix}$ and consider a time series $\bold X$ with transition probabilities $a_{ij}(1)=Tr(\ket{j} \bra{j} Tr_{\mathcal{H}_*} (U \ket{i} \bra{i} \otimes \ket{0} \bra{0} U^{\dagger}))$, so $\bold X$ is $K$-coherent by definition and by the previous theorem. Note that by the explicit description of $a_{ij}(k)$ that we obtained in the beginning of this section, it follows that $0\leq a_{ij}(1)$ for all $i, j \in \{0, 1\}$ and it is easy to see that $a_{i0}(1)+a_{i1}(1)=1$ ($i\in \{0, 1\}$), so the above are indeed probability distributions. Suppose that $\bold X$ can be modelled by a closed system, then there is a unitary $V$ such that $\bra{j} V \ket{i}=a_{ij}(1)$ for all $i, j \in \{0, 1\}$. Therefore, by the characterization of the $a_{ij}(1)$ that we previously obtained,
\begin{align}
V=\begin{pNiceMatrix} a_{00}(1) & a_{10}(1) \\ a_{01}(1) & a_{11}(1) \end{pNiceMatrix} = \begin{pNiceMatrix} \frac{1}{2} & \frac{1}{2} \\ \frac{1}{2} & \frac{1}{2} \end{pNiceMatrix}.
\end{align}
But $V$ is not unitary, and so we get a contradiction. It follows that $\bold X$ cannot be modelled by a closed system. 
\end{proof}

We conclude with the following conjecture.

\begin{conjecture}\label{conject}
There exists a time series $\bold X$ and some $K$ such that $\bold X$ is not $K$-Coherent. 
\end{conjecture}

\section{On the hardness of finding the underlying unitary}\label{ssec:hardness}

We continue working under the same assumptions --- all entries of U are real numbers and $U_k=U^k$. Consider the following problem:
\begin{flushleft}
\textbf{Input}: A set of natural numbers $K=\{k_1,...,k_{n-1}\}$, and for each $k\in K$, we are given transition probabilities $\{T_{ij}(k) : i, j \in \{0, 1\}\}$, such that the probabilities satisfy the requirements of $K$-Coherence.
\end{flushleft}
\begin{flushleft}
\textbf{Output}: A $4\times 4$ unitary $U$ (whose entries are enumerated $(b_1,...,b_{16})$) such that, for each $k\in K$, we have the following:
\end{flushleft}

\begin{align}
T_{00}(k)&=(p_{11}^k p_{11}^k +p_{21}^k p_{21}^k)(b_1,...,b_{16})\\
T_{01}(k)&=(p_{31}^k p_{31}^k+ p_{41}^k p_{41}^k)(b_1,...,b_{16})\\
T_{10}(k)&=(p_{13}^k p_{13}^k+p_{23}^k p_{23}^k)(b_1,...,b_{16})\\
T_{11}(k)&=(p_{33}^k p_{33}^k+p_{43}^k p_{43}^k)(b_1,...,b_{16}).
\label{eqn:U_bs}
\end{align}

By the previous section, solving the above problem is equivalent to finding the underlying unitary that witnesses $K$-coherence. By ADD REFERENCE, the problem of solving general such systems is in $PSPACE$, and under the Generalized Riemann Hypothesis it's in $AM$ (which still contains $NP$) and no better bounds are known at the moment. As we shall see in the next section, we will be able to find a reasonably small $\epsilon$ and a unitary $U$ whose entries will satisfy:
\begin{align}
|T_{00}(k)-(p_{11}^k p_{11}^k +p_{21}^k p_{21}^k)(b_1,...,b_{16})|&<\epsilon\\
|T_{01}(k)-(p_{31}^k p_{31}^k+ p_{41}^k p_{41}^k)(b_1,...,b_{16})|&<\epsilon\\
|T_{10}(k)-(p_{13}^k p_{13}^k+p_{23}^k p_{23}^k)(b_1,...,b_{16})|&< \epsilon\\
|T_{11}(k)-(p_{33}^k p_{33}^k+p_{43}^k p_{43}^k)(b_1,...,b_{16})|&<\epsilon.
\label{eqn:U_epsilon}
\end{align}

Following this, we may consider the notion of $(\epsilon, K)$-Coherence for $\epsilon>0$, where in the definition of $K$-Coherence, we replace equations~\autoref{eqn:U_bs} with equations~\autoref{eqn:U_epsilon}. The general problem of finding arbitrarily good approximations for solutions of systems of polynomial equations with integer coefficients can be shown to be $NP$-hard using a reduction from $3$-$SAT$. This leaves several questions open:

\begin{itemize}
\item  Does $(\epsilon, K)$-Coherence for arbitrarily small $\epsilon$ imply $K$-Coherence?

\item Is every time series $(\epsilon, K)$-Coherent for an appropriate $K$ and arbitrarily small $\epsilon$'s? Our experiments may be taken as an indication that the answer might be in the affirmative for a large class of time series.
\end{itemize}

\section{Approximating Canonical Dynamics}\label{ssec:canonicald}
\begin{figure}
\includegraphics[width=0.6\linewidth]{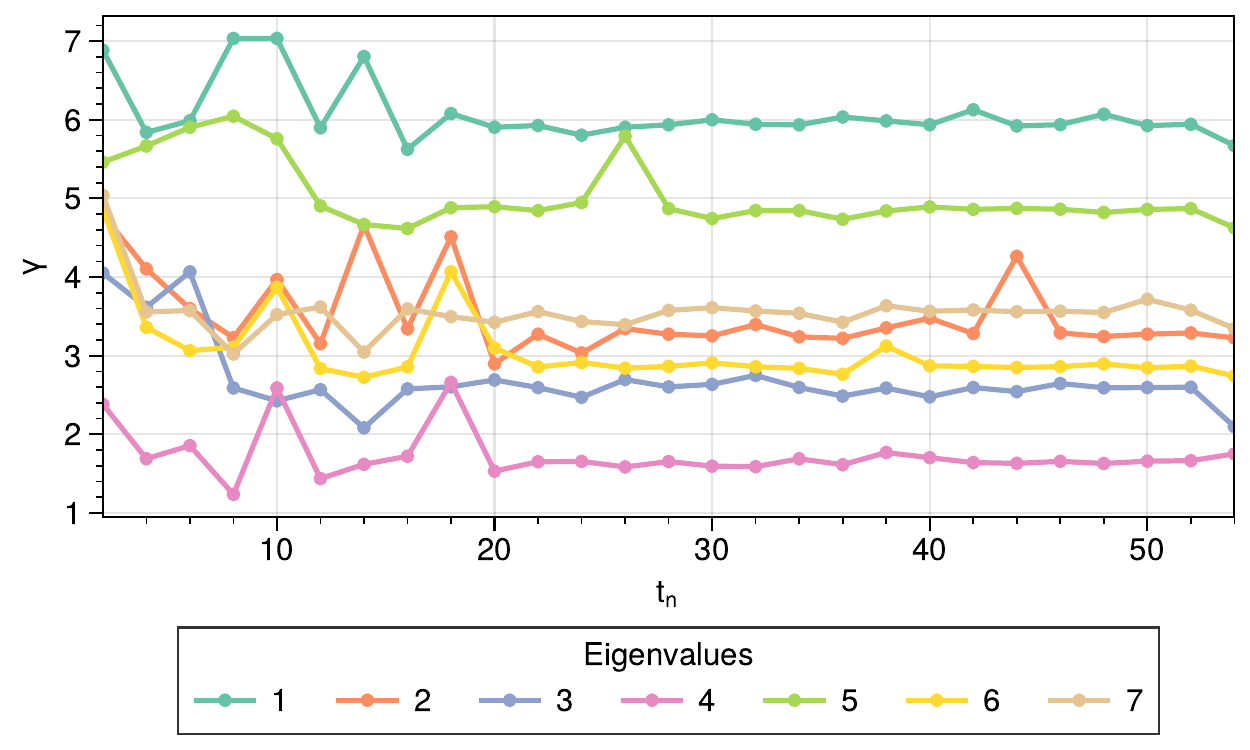}
    \caption{$\gamma$ as a function of $t_n$. $\gamma$ is indicated in \autoref{eq:ans}, which at $t=1$ are proportional to the eigenvalues of the learnt Hamiltonian. $t_n$ refers to the number of transition matrices being fed into the learning process with time steps at each point in $x-$axis having $K\in {1,2,4,...2t_n}$.   }
    \label{fig:convergence}
\end{figure}

For a natural $n$, let $K_n=\{1,...,n\}$. Note that for every set $K$ of natural numbers and for every collection of real numbers $P=\{a_{ij}(k): i, j \in \{0, 1\}, k\in K \}$ (not necessarily transition probabilities of a $K$-Coherent time series), the set of unitaries whose entries satisfy the polynomial equations of the previous section~\autoref{eqn:poly} for all $k\in K$ is compact. Denote this set by $S_K^P$. Suppose that $K\subseteq K'$ and that $P=\{a_{ij}(k): i, j \in \{0, 1\}, k\in K \}$ and $P'=\{a_{ij}'(k): i, j \in \{0, 1\}, k\in K' \}$ are given such that $a_{ij}(k)=a_{ij}'(k)$ for every $k\in K$, then $S_{K'}^{P'} \subseteq S_K^P$.

Let $U$ be a $4\times 4$ unitary as in the previous section. Fix a large $n$, and for each $k\in K_n$, consider the transition probabilities $T_{ij}(k)$ induced from the $CPTP$ map $\mathcal{E}_k^U$ as in the previous section. Let $P_n=\{T_{ij}(k) : i, j \in \{0, 1\}, k\in K_n\}$, and for every $m<n$, let $P_m=\{T_{ij}(k) : i,j \in \{0,1\}, k\in K_m\}$. Then $m<l \leq n \rightarrow S_{K_l}^{P_l} \subseteq S_{K_m}^{P_m}$. As the sets $S_{K_i}^{P_i}$ are compact, the distance $d_m:=d(S_{K_m}^{P_m}, S_{K_n}^{P_n})$ is well-defined for all $m\leq n$ and $d_m$ is decreasing with $m$.

Since the function that extracts transition probabilities $\{T_{ij}(k) : k\in K\}$ from a unitary as in the previous section is continuous, we should expect that as $m<n$ increases, the unitary in $S_{K_m}^{P_m}$ that we find will give rise to transition probabilities that approach $P_n$.

Furthermore, we conducted the following experiment, which suggests a similar type of  ``stabilization" phenomenon as we train with more time steps. From a given time series, for every $n$ we consider the learning problem $P_n$ as above, and so for each $P_n$, we increase the number of time steps we are exposing to the quantum network. Intuitively, if our assumption of the processes being emitted by a single quantum process along with our ability to model that process is true, then the resulting learned Hamiltonian would converge as we increase $t_n$. To this end, in \autoref{fig:convergence} we look at the diagonal elements of the learned unitary, which are proportional to the eigenvalue of the learned Hamiltonian. We see that the eigenvalues are convergent indicating that as more and more time steps (transition matrix) are involved in the learning process, the closer we get to a ``canonical" set of eigenvalues. This provides some evidence for the existence of a canonical $\bar{H}$ that underlies the process.

\section{Learning Procedure}\label{ssec:learn}

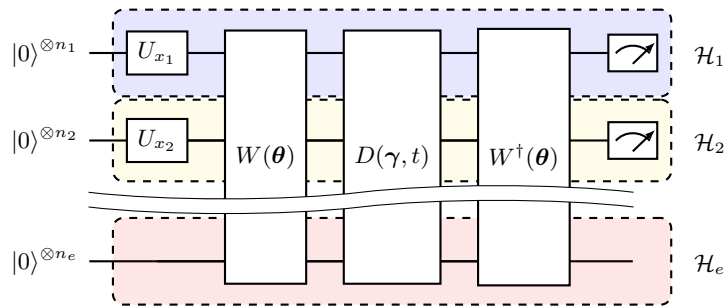
\begin{figure}[H]
    \centering
\begin{quantikz}
\lstick[wires=1]{$\ket{0}^{\otimes n_1}$}& \gate[wires=1]{U_{x_1}}\gategroup[1,steps=5,style={dashed,
rounded corners,fill=blue!10, inner xsep=2pt},background,label style={label position=right,anchor=
north,xshift=0.5cm}]{ $\mathcal{H}_1$} & \gate[wires=4]{W(\vect{\theta})} & \gate[wires=4]{D(\vect{\gamma},t)}&\gate[wires=4]{W^\dagger(\vect{\theta})}& \meter{} \\
\lstick[wires=1]{$\ket{0}^{\otimes n_2}$}& \gate[wires=1]{U_{x_2}}\gategroup[1,steps=5,style={dashed,
rounded corners,fill=yellow!10, inner xsep=2pt},background,label style={label position=right,anchor=
north,xshift=0.5cm}]{ $\mathcal{H}_2$}  & \qw& \qw& \qw & \meter{} \\
\wave&&&&&\\
\lstick[wires=1]{$\ket{0}^{\otimes n_e}$}& \qw\gategroup[1,steps=5,style={dashed,
rounded corners,fill=red!10, inner xsep=2pt},background,label style={label position=right,anchor=
north,xshift=0.5cm}]{ $\mathcal{H}_e$}& \qw& \qw& \qw& \qw
\end{quantikz}
\caption{Ansatz used as time-generative model, We start with loading our initial states in their respective sub-spaces, followed by applying the parameterized Hamiltonian evolution. Finally we measure the individual sub-spaces (marked $\mathcal{H_i}$), the probability of which is used for calculating the loss function. }\label{fig:full_ansatz}
\end{figure}

Our goal is to learn a stationary quantum Hamiltonian, $\operatorname{H}$, whose unitary time evolution generates the given time series. The corresponding unitary is denoted by $U = e^{-i\operatorname{H}t}$. In this section, we set up infrastructure for the learning procedure and test the general learning ability of the algorithm before applying it to our use case in the next section. The first step in the algorithm is to encode the initial state using the basis encoding. The aforementioned Hamiltonian, $\operatorname{H}$, is learned via training a variational circuit with appropriately chosen circuit ansatz and parameters. These parameters are trained iteratively by the variational algorithm to minimize the chosen cost function.

One of the key aspects of this process is choosing an appropriate ansatz. We follow the ansatz introduced in Ref.~\cite{qcl}. This was used in Ref.~\cite{qcl} to learn the Hamiltonian responsible for constraining the search space, while here we learn the Hamiltonian that reproduces the given time dynamics. This is shown in \autoref{fig:full_ansatz}. We start by first choosing the number of qubits/discretization required for each subspace in our multi-dimensional model, which dictates the number of qubits ($n_i$) required for representing each time series. We also choose the appropriate dimensions for the environment space - $n_e$. We then encode the basis states ($x_i$s)\footnote{$x_i$ are the terms that enter the cost function as the $i^{th}$ index in $a^{lh}_{ijkd}$ in main paper.} we start with in their respective subspace using $U_{x_i}$, by using a basis embedding circuit. We then Apply the ansatz  $V_{lk}(\vect{\alpha}, t)$, which is given by

\begin{equation*}
V_{lh}(\vect{\alpha}, t) := W_l(\theta)D_h(\gamma, t)W_l(\theta)^{\dagger} = e^{-i\mathcal{H}(\theta,\gamma)t}, \label{eq:ans}
\end{equation*}

where $W_l(\theta)$ is a $l$-layered parameterized quantum circuit, $D_h(\gamma, t)$ is a $h$-local parameterized diagonal unitary (explained in following paragraphs) evolved for time $t$ and $\vect{\alpha} = \{\gamma, \theta\}$ is the parameter set (as shown in \autoref{fig:full_ansatz}).

$V_{lh}(\vect{\alpha}, t)$ could be better understood by following the application order of unitaries. Once the basis embedding is applied, we get a state given by $\ket{\psi_1}=\ket{x_1}\otimes\ket{x_2}\otimes\ldots\ket{0}^{\otimes{n_e}}$. Further application of the unitary $W^{\dagger}$, leads to the state $\ket{\psi_2}=W^{\dagger}\ket{\psi_1}$, which is equivalent to a parametrized rotation to a new basis set. We hope here that the perfect learning parameters leads to a choice of $W$, which would rotate it to the diagonal basis of the unknown Hamiltonian we are trying to learn. We then evolve the state $\ket{\psi_2}$ for a given time $t$ with the eigenvalues of the Hamiltonian given by $\vect{\gamma}$ - $\ket{\psi_3}=D(\vect{\gamma},t)\ket{\psi_2}$. Since our hope is to learn the ideal eigenbasis of this Hamiltonian, we gain the ability to time evolve the system by simply multiplying the diagonal eigenvalue matrix with the corresponding time step. Finally, to measure the final state in the original basis, we revert back to it by applying $\ket{\psi_4}=W \ket{\psi_3}$. Thus the transformation from $\ket{\psi_1}\rightarrow\ket{\psi_4}$ involves evolving the initial state for an appropriate time with a parameterized hermitian Hamiltonian. Our goal is to find such a Hamiltonian, partial trace of which gives us the correct statistics to reach the various states in various sub-spaces we .

More concretely, for the numerical experiments, $W$ is implemented as shown in \autoref{fig:W_ansatz}. It contains $X,Y,Z$ single qubit parameterized rotation gates for each qubit followed by $CNOT$ gates that entangle every qubit with its $l^{th}$ successor at $l^{th}$ layer \textit{i.e.} CNOT gates between the qubits $j$ and $(j + i)\%n \quad j \in [1, n]$ where $i$ is the number of the layer. The diagonal unitary $D_h$ is prepared efficiently using a Walsh–Fourier series approximation as done in \cite{Welch2014}. It is given by

\begin{equation}
    D = e^{iG} = \prod\limits_{j=0}^{2^q-1} e^{i\gamma_j}\otimes_{m=1}^{n}(Z_m)^{j_m},
    \label{eq:diagonal_unitary}
\end{equation}

where $q = n$, $G$ and $Z_m$ are diagonal operators with Pauli operator $Z_m$ acting on $m$-th bit in bitstring $j$. Efficient quantum circuits for minimum depth approximations of $D$ is obtained by resampling the function on the diagonal of $G$ at sequences lower than a fixed threshold, with $q = m$, with $m <= n$. The resampled diagonal takes the same form as \autoref{eq:diagonal_unitary} but with q = k~\cite{Welch2014}.

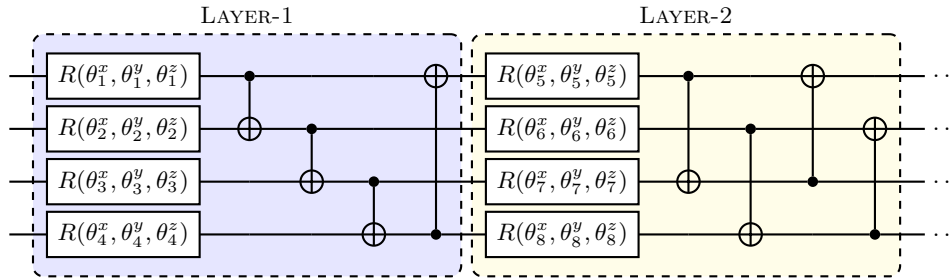
\begin{figure}[H]
\centering
\begin{quantikz}[row sep=0.1cm]
& \gate{R(\theta^x_1,\theta^y_1,\theta^z_1)}\gategroup[4,steps=5,style={dashed,rounded corners,fill=blue!10, inner xsep=2pt},background]{{\sc Layer-1}}        
& \ctrl{1} & \qw      & \qw      &  \targ{} & \gate{R(\theta^x_5,\theta^y_5,\theta^z_5)}\gategroup[4,steps=5,style={dashed,rounded corners,fill=yellow!10, inner xsep=2pt},background]{{\sc Layer-2}} 
                                                                                                                                        & \ctrl{2}  & \qw       & \targ{}   &  \qw & \ \ldots \qw  \\
& \gate{R(\theta^x_2,\theta^y_2,\theta^z_2)}  & \targ{}  & \ctrl{1} & \qw      & \qw      & \gate{R(\theta^x_6,\theta^y_6,\theta^z_6)} & \qw       & \ctrl{2}  & \qw       & \targ{}& \ \ldots \qw \\
& \gate{R(\theta^x_3,\theta^y_3,\theta^z_3)}        & \qw      & \targ{}  & \ctrl{1} & \qw      & \gate{R(\theta^x_7,\theta^y_7,\theta^z_7)} & \targ{}   & \qw       & \ctrl{-2} & \qw & \ \ldots \qw     \\
& \gate{R(\theta^x_{4},\theta^y_{4},\theta^z_{4})}  & \qw      & \qw      & \targ{}  & \ctrl{-3}& \gate{R(\theta^x_8,\theta^y_8,\theta^z_8)}  & \qw       & \targ{}   & \qw       & \ctrl{-2}& \ \ldots \qw       
\end{quantikz}
\caption{Ansatz for $W(\vect{\theta})$ used for numerical experiments, where in each layer, we have 3 single qubit parameterized rotation gates $X,Y,Z$ along with $CNOT$ entangling the qubits at distance $l$ apart for $l^{th}$ layer.}\label{fig:W_ansatz}
\end{figure}

The PQC is a junction between the quantum and classical computations, where the probability distribution obtained from repeated measurements of the quantum variational circuit are fed into the classically evaluated cost function. The parameter values are refined by the optimizer and fed back into the variational circuit until a desired level of accuracy is reached. Throughout the paper, as mentioned in main text, we have used the COBYLA~\cite{powell1994direct} optimization scheme for learning the similarity. Because of the costly function evaluations required to compute the terms in loss function, we apply similar heuristics introduced in Ref.~\cite{https://doi.org/10.48550/arxiv.2201.02310}, where the concept of stochastic batching was used to reduce the number of evaluations. We pick $N_{\text{batch}}$ random points from the set of points made up of time step and the initial basis set from the training set for $i^{th}$ iteration. We use this batch cost function in the COBYLA optimization routine until convergence is achieved. We note that this form of ``stochastic'' COBYLA is similar to stochastic gradient descent techniques~\cite{bottou2010large}, but has no theoretical backing yet in the literature.

\subsection{Trainablity of ansatz}
 \begin{figure}[htbp]   
    \centering
    \includegraphics[width=0.6\linewidth]{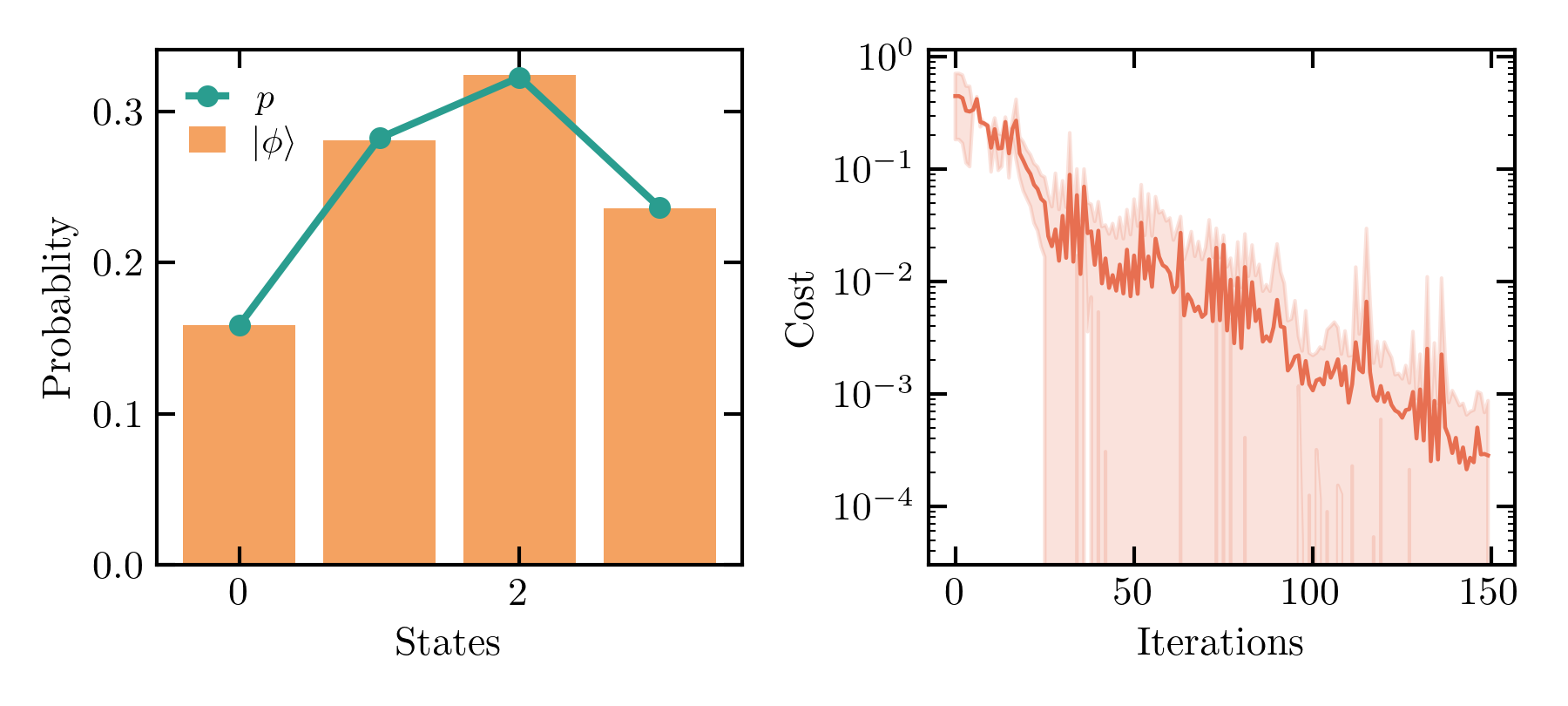}
    \caption{Learnability of a general probability distribution by the given ansatz. \textit{(left)} Example of a learned distribution, \textit{(right)} Average cost of learning 50 random distributions}
    \label{fig:learnablity}
\end{figure}
Before applying it on the problem of interest, one cannot easily guarantee that a given ansatz can learn a general distribution. It was shown in Ref.~\cite{qcl} that ansatz of the form used here can indeed learn distributions that gives either equal superposition of a set of basis set or a single basis set. Here we extend that and numerically show that the anstaz is capable of learning general probability distributions (albiet with error due to expressivity). This is important as our goal is indeed to learn a general distribution. To this end, we generate $10^3$ random $n=2$ qubit probability distribution, which we try to learn with an initial starting state $U_{x_1}|_{x_1=1}\ket{00}=\ket{11}$. That is, we try to learn a Hamiltonian that will evolve the state $\ket{11}$ in time $t$ (again randomly sampled) to a state which gives us a probablity distribution $p_i$ for $10^3$ different $p_i$s. We here use fully $2-$local diagonal ansatz (fully connected) and $2$ layers for $W$ ansatz. In \autoref{fig:learnablity}(right), we show one such $p_i$ (green line), along with the learned distribution (orange histogram) evolved using the trained Hamiltonian. We also show in \autoref{fig:learnablity}(left) the mean and standard deviation of cost at each iteration for all $10^3$ runs.
 
After the training procedure, the model can be used in various ways to generate the time series data. To generate the time series data at a particular time step $k$ for this study, we run the model with the learned transition matrix $T(k)$ and use the state from the single-shot measurement without any post processing. So, to generate a series of length m, we will need to run the circuit with $T(k)$ with $k=1,2,..$ respectively and use the states as the generated time series data from the measurements of each of these runs. However, one can use various post-processing techniques. One such technique is to use multiple shots, say m, per time step and then using the probability distribution obtained from these measurements. This obviously increases the quantum complexity from 1 to m shots per time step. One can also use the single-shot data in conjunction with maximum likelihood estimation to generate a probability distribution of states instead of a single state.

\section{Hamiltonian learning, Non-Markovianity and complexity}\label{ssec:nonmarkov}

\begin{figure}[!htb]
   \centering
   \includegraphics[width=0.6\linewidth]{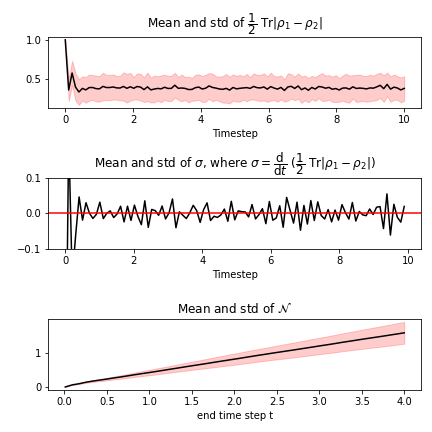}
    \caption{Characterization and quantification of non-Markovianity using $D$ (\textit{top}), $\sigma$ (\textit{middle}) and $\mathcal{N}$ (\textit{bottom}) as defined in \autoref{eq:D}, \autoref{eq:sigma}, \autoref{eq:N} respectively} 
    
    \label{fig:measure}
\end{figure}

As noted by Feynman~\cite{feynman1982simulating}, simulating quantum Hamiltonians requires exponentially large resources as the system scales up. In reality, one usually resolves to perturbation theory which can only be applied to problems that can be split into an exactly solvable part and a perturbative part, provided that a relevant small parameter (e.g., weak coupling strength with respect to other energy scales) can be established. It is well known that even the static limit of strongly interacting quantum systems are often inaccessible analytically. Although numerical techniques for such systems like Quantum Monte Carlo (QMC)~\cite{landau2021guide} have been introduced, it is often plagued with stability problems. Things turn quickly out of hands when one starts exploring time dynamics of such strongly correlated quantum systems~\cite{chen2017inchworm,prosen2007efficiency}. Although there exists a class of exactly solvable models of open quantum dynamics \cite{filippov2017divisibility}, without the Markov approximation, such problems is typically impossible to solve directly because of the exponentially large dimension of the environment’s Hilbert space~\cite{holevo2019quantum}. Thus, when one starts dealing with non-Markovian dynamics, things become even more harder requiring very specific assumptions to simulate time dynamics such systems. In following paragraphs, we show that the stable Hamiltonian that one learns for given classical series indeed resides in this non-Markovian regime. This is not a rigorous classification of complexity for simulating the given classical process, instead we just analyze the complexity of simulating the learnt equivalent quantum process. Thus to conclude, we note that the non-Markovian dynamics are hard to model classically and therein lies the potential quantum advantage.

 We use the metric proposed in Ref.~\cite{Breuer2009} as a measure of Markovianity. Briefly, a Markovian process tends to diminish the distinguishability between two states that start out maximally distinguishable, while a non-Markovian process retains some degree of distinguishability over time. Ref.~\cite{Breuer2009} uses this notion of distinguishability as a proxy for the measure. To this end, one uses the trace distance as a distance measure between two quantum states, given by

\begin{equation}
    D = \dfrac{1}{2}~Tr\left|\rho_1 - \rho_2 \right|,\label{eq:D}
\end{equation}

where $\rho_1$ and $\rho_2$ are two quantum states and $|A| = \sqrt{A^\dagger A}$. Trace distance tends to be a very natural metric in density matrix space with the property $0\leq D \leq 1$. Moreover, one can show that no trace preserving quantum operation can ever increase the trace-distance of two states~\cite{ruskai1994beyond}, which is the exact property we need in our measure.

As the system evolves given by the rules of the quantum process, one then looks at the rate of change of this distance w.r.t. time for the states $\rho_1,\rho_2$, evolved with a given quantum process $\Phi(t)$, which is give by 
\begin{equation}
    \sigma (t, \rho_{1,2}(0)) = \diff{}{t} D(\rho_1(t),\rho_2(t)).\label{eq:sigma}
\end{equation}

This quantity depends not only on time $t$ but also on the initial states $\rho_{1/2}(0))$ as $\rho_{1/2}(t)=\Phi(t, 0) \rho_{1/2}(0)$. It was shown in Ref.~\cite{Breuer2009} that $\sigma(t)\leq 0$  for all quantum
Markovian processes. Thus, when a given process attains $\sigma(t)> 0$ for some time $t$, the process is said to be non-Markovian during such time, with the degree of non-Markovianity given by $\sigma(t)$. Correspondingly to characterize the process, one can look at the following cumulative measures over time

\begin{equation}
   \mathcal{N}(\Phi,t) = \max_{\rho_{1,2}(0)} \int_{0}^t dt~\sigma(t, \rho_{1,2}(0))|_{\sigma > 0}.\label{eq:N}
\end{equation}
$\mathcal{N}(\Phi,t)$ measures the cumulative degree of non-Markovianity of the given process up until time $t$. Although we use the above measure, there are equivalent and alternative measures~\cite{Rivas2014}, with some researchers employing machine learning techniques to quantify such effects~\cite{Fanchini2021}.

We start by using the learned parameter for one dimensional time series. We choose various starting states $\rho_{0/1}$ for initialization (seen as mean/spread in figure). We then compute the various quantities we are interested in. In the top panel of \autoref{fig:measure}, we plot the metric $D$ that measures the difference at each time-step between the time series evolved by the learned quantum process. Although the states start with being maximally distinguishable, they quickly lose the distinctness. But looking closer at the rate of change of such measure (\autoref{fig:measure} \textit{middle}), we see that there exists time steps where the states move apart from each other (indicated by $\sigma >0$). Interestingly, when one looks at the cumulative measure, we see that non-Markovianity is injected more and more (in fact linearly) as the system is evolved. As a result of this increasing (non-saturating)  non-Markovianity, the system gets harder and harder to be simulated classically.

\bibliography{refs}

\end{document}